\documentclass{article}

\usepackage[preprint]{neurips_2026}

\usepackage[utf8]{inputenc}
\usepackage[T1]{fontenc}
\usepackage{url}
\usepackage{booktabs}
\usepackage{amsfonts}
\usepackage{nicefrac}
\usepackage{microtype}
\usepackage{xcolor}
\usepackage{graphicx}
\usepackage{amsmath, amssymb, amsthm}
\usepackage{enumitem}
\usepackage{placeins}
\usepackage{algorithm}
\usepackage{algpseudocodex}
\usepackage{hyperref}

\newtheorem{definition}{Definition}
\newtheorem{proposition}{Proposition}

\newcommand{\tuple}[1]{\ensuremath{\langle{#1}\rangle}}

\usepackage[textsize=tiny,
    ]{todonotes}

\usepackage[normalem]{ulem}

\newcommand{\AgentSet}{A}
\newcommand{\NumAgents}{N}
\newcommand{\StateSpace}{\mathcal{S}}
\newcommand{\ActionSpace}[1]{\mathcal{A}^{#1}}
\newcommand{\TransFunc}{\mathcal{T}}
\newcommand{\RewardFn}{r}
\newcommand{\Discount}{\gamma}

\newcommand{\PartitionVar}{P}
\newcommand{\PartitionSet}{\mathcal{P}}
\newcommand{\NumTeams}{K}
\newcommand{\TeamVar}[1]{T_{#1}}
\newcommand{\GoalVar}[1]{g_{#1}}
\newcommand{\GoalSpace}[1]{\mathcal{G}_{#1}}
\newcommand{\CandidateGoalSet}{\mathcal{G}}
\newcommand{\HypothesisVar}{h}
\newcommand{\DomainModel}{D}

\newcommand{\ObsSeq}[1]{O_{1:#1}}
\newcommand{\StateVar}{s}
\newcommand{\JointActionVar}{a}
\newcommand{\Policy}{\pi}
\newcommand{\PolicyParams}{\theta}
\newcommand{\PolicyTheta}{\Policy_{\PolicyParams}}
\newcommand{\ActCtx}{c}

\newcommand{\CounterObs}[3]{o_{#1}^{#2}\!\left(#3\right)}
\newcommand{\ScoreFn}[2]{S_{#1}\!\left(#2\right)}
\newcommand{\UpperBoundFn}[2]{U_{#1}\!\left(#2\right)}
\newcommand{\StaleUpperBoundFn}[2]{\widetilde{U}_{#1}\!\left(#2\right)}
\newcommand{\Floor}{F}
\newcommand{\TopKHeap}{\mathcal{H}}
\newcommand{\RankLimit}{k}

\newcommand{\GoalTupleQueue}{\mathcal{Q}_{\mathrm{goal}}}

\newcommand{\TermPenalty}{\lambda_{\mathrm{term}}}
\newcommand{\ScoreEps}{\varepsilon_{\mathrm{score}}}

\newcommand{\StateFeats}[1]{\phi_{\texttt{state}}\!\left(#1\right)}

\newcommand{\GoalEmbed}[1]{\phi_{\texttt{goal}}\!\left(#1\right)}
\newcommand{\SelfEmbed}[1]{e_{\texttt{self}}\!\left(#1\right)}
\newcommand{\TeamEmbed}[1]{e_{\texttt{team}}\!\left(#1\right)}

\title{Multi-Agent Goal Recognition with Team- and Goal-Conditioned Reinforcement Learning and Factorized Branch-and-Bound}

\author{%
  Thiago Thomas \\
  Pontif\'{i}cia Universidade Cat\'{o}lica do Rio Grande do Sul (PUCRS) \\
  \And
  Gabriel de Oliveira Ramos \\
  Universidade do Vale do Rio dos Sinos (Unisinos) \\
  \And
  Felipe Meneguzzi \\
  University of Aberdeen \\
}

\begin{document}

\maketitle

\begin{abstract}
Multi-agent goal recognition asks an observer to jointly infer which agents act together and what each team is trying to achieve, so the hypothesis space grows combinatorially with the number of team partitions and goals per team.
Real applications such as drone surveillance and collaborative robotics expose only the agents' trajectory, which forces the observer to rank team-goal hypotheses from behavior alone.
Multi-Agent Goal Recognition with Branch-and-Bound (\emph{MAGR-BB}) addresses this setting with a shared team- and goal-conditioned policy used as the scoring model inside a factorized branch-and-bound search.
On a controlled multi-agent Blocksworld benchmark, MAGR-BB returns the same top-ranked hypothesis as exhaustive search throughout the trajectory while cutting hypothesis materialization by orders of magnitude and reducing cumulative recognition runtime substantially.
\end{abstract}

\section{Introduction}
\label{sec:introduction}
Goal recognition infers hidden intent from observed behavior and supports assistive robotics, security, and human-AI teaming~\citep{kautz1986generalized,demiris2007robotics,sukthankar2014plan,mirsky2021introduction,meneguzzi2021survey}.
In multi-agent settings, intent is shared across actors: an observer sees several agents moving simultaneously and must infer both which agents coordinate and what goal each resulting team pursues.
This distinction matters for drone-attack recognition~\citep{abdelkader2021aerial}, cybersecurity~\citep{mirsky2019cybersecurity}, human intent recognition~\citep{singh2020application}, autonomous driving~\citep{albrecht2021autonomousDriving}, and multi-robot coordination~\citep{farinelli2004multirobot}.
For example, two drones flying toward adjacent buildings may be surveying the area together, pursuing two separate goals, or coordinating an attack.

Most goal recognizers still target a single actor.
Planning-based recognizers~\citep{ramirez2009plan,ramirez2010probabilistic,pereira2020landmark} and the single-agent model-free line, including GRAQL~\citep{amado2022graql}, the deep-RL recognizer of \citet{Fang2023}, and GRNet~\citep{grNet}, all infer one actor's hidden goal at a time.
The multi-agent recognizers closest to our setting use explicit plan libraries, action models, or BDI intention structures~\citep{banerjee2010multiagent,zhuo2012action,argenta2015maprap,argenta2016probabilistic,zhuo2019recognizing,michaeldann2023multiagent}.
To the best of our knowledge, recognition with learned behavior policies therefore lacks a method that recognizes hidden teams and team goals together.

The problem is combinatorial.
An observer must jointly infer who coordinates with whom and what each team is trying to achieve.
\citet{banerjee2010multiagent} show that multi-agent plan recognition is NP-complete when the number of agents is variable.
Let $\PartitionSet$ be the candidate partitions of agents into $\NumTeams$ labeled team slots, and let $\GoalSpace{k}$ be the candidate-goal library for slot $k$.
An exhaustive baseline then enumerates $\sum_{\PartitionVar\in\PartitionSet}\prod_{k=1}^{\NumTeams}|\GoalSpace{k}|$ complete partition-goal hypotheses at each observed step.
In our benchmark, this count is $7{,}154{,}784$ at each observed step.

In this paper, we introduce MAGR-BB, a model-free recognizer for jointly inferring hidden team partitions and team goals from a fully observed joint trajectory.
MAGR-BB combines the single-agent learned-policy scoring of GRAQL~\citep{amado2022graql} with partition-level bounds inspired by coalition-structure search~\citep{rahwan2007coalition} under a \emph{non-competitive} score: each team-goal score depends only on its own team and goal, independent of the goals assigned to other teams in the partition.
Because team scores add to the score of a complete hypothesis, the recognizer caches local team-goal scores, prunes partitions whose best total cannot beat the lowest score in the current top-$k$ list, and preserves the exhaustive top-$k$ ranking without constructing every complete hypothesis.

\noindent\textbf{Contributions:}
\begin{enumerate}
  \item We formalize multi-agent goal recognition as inferring which agents form teams and which goal each team pursues from a fully observed joint trajectory (Section~\ref{sec:method:problem}).
  \item We train a single Transformer policy~\citep{vaswani2017attention} conditioned on a candidate goal and candidate team to score any candidate team-goal pair (Section~\ref{sec:method:policy}), extending the per-goal model in GRAQL~\citep{amado2022graql} to a single shared network across all hypotheses.
  \item We derive MAGR-BB, a branch-and-bound recognizer whose team-level bounds over partitions and goal assignments preserve the complete-hypothesis ranking under the non-competitive scoring condition (Section~\ref{sec:method:recognition}).
  \item We empirically show on a controlled multi-agent Blocksworld benchmark~\citep{slaney2001blocksworld} that MAGR-BB returns the same top-ranked hypothesis at every observed step and the same final top-$10$ list as the exhaustive baseline, builds only $10$ complete hypotheses rather than $7.15$M at the final observed step, and cuts cumulative runtime by a factor of $2.43$--$2.91$ (Section~\ref{sec:experiments}).
\end{enumerate}

\section{Preliminaries}
\label{sec:preliminaries}
\subsection{Multi-Agent Decision Processes}
\label{sec:preliminaries:madp}

We use the standard cooperative multi-agent MDP tuple $\tuple{\StateSpace,\{\ActionSpace{i}\}_{i=1}^{\NumAgents},\TransFunc,\RewardFn,\Discount}$, where $\StateSpace$ is the state space, $\ActionSpace{i}$ is agent $i$'s action space, $\TransFunc$ is the transition function, $\RewardFn$ is the shared team reward, and $\Discount$ is the discount factor.
A joint action $\JointActionVar=(\JointActionVar^1,\dots,\JointActionVar^\NumAgents)$ drives the shared environment~\citep{albrecht2024autonomous}.
We train under centralized training with decentralized execution (CTDE), where the critic observes the joint state during learning while each actor uses only its local observation at execution time~\citep{lowe2017maddpg,albrecht2024autonomous}.
At recognition time, the observer receives fully observed joint state-action pairs; we leave partial observability for future work.

\subsection{Goal-Conditioned Reinforcement Learning}
\label{sec:preliminaries:rl}

We train a goal-conditioned policy~\citep{schaul2015uvfa,liu2022goalconditioned} with PPO~\citep{schulman2017ppo}, which optimizes policy-gradient updates through a clipped surrogate objective.
The candidate goal enters the policy input, and we train under a length-based curriculum~\citep{bengio2009curriculum} that gradually increases goal length and difficulty.

\subsection{Plan and Goal Recognition}
\label{sec:preliminaries:planning}

A classical planning problem consists of a state space, action schemas with preconditions and effects, an initial state, and a goal condition~\citep{ghallab2004planning,geffner2013planningbook}.
Plan and goal recognition invert plan generation by mapping observed behavior back to a likely plan or goal~\citep{kautz1986generalized,sukthankar2014plan,mirsky2021introduction,meneguzzi2021survey}.
Given a domain model $\DomainModel$, candidate goals $\CandidateGoalSet$, and observations $\ObsSeq{t}$, a recognizer produces a ranking over $\CandidateGoalSet$.
Existing families include plan-library-based methods~\citep{geib2009plan}, planning-based methods~\citep{ramirez2009plan,ramirez2010probabilistic,pereira2020landmark}, and model-free methods that replace the domain model with learned value functions, policies, or neural predictors~\citep{amado2022graql,Fang2023,grNet,amadosurveymodelfree}.

\subsection{Branch-and-Bound}
\label{sec:preliminaries:bnb}

Branch-and-bound is a combinatorial search framework that alternates between branching over subproblems and pruning any whose admissible bound cannot beat the current incumbent~\citep{lawler1966branch}.
Every pruned subtree has an upper bound no larger than the best retained solution, so the procedure returns the same top-ranked solutions as exhaustive search.

\section{MAGR-BB Method}
\label{sec:method}

This section first formalizes the recognition problem.
It then defines the learned policy, the scoring rule, and the branch-and-bound recognizer.

\subsection{Problem Formulation}
\label{sec:method:problem}

Let $\AgentSet=\{1,\dots,\NumAgents\}$ be the set of agents.
Let $\PartitionVar=\{\TeamVar{1},\dots,\TeamVar{\NumTeams}\}$ be a hidden indexed partition of those agents into $\NumTeams$ teams, with the index $k$ acting as a team-slot label that distinguishes otherwise equivalent partitions.
The teams are disjoint and satisfy $\bigcup_k \TeamVar{k}=\AgentSet$.
For each labeled team slot $k$, $\GoalSpace{k}$ is a finite set of candidate goal specifications that the policy can condition on.
The hidden team goal $\GoalVar{k}$ for team $\TeamVar{k}$ is one element of $\GoalSpace{k}$.
A complete recognition hypothesis is therefore
\[
\HypothesisVar = (\PartitionVar, \GoalVar{1},\dots,\GoalVar{\NumTeams}).
\]
We assume full observability of the recorded trajectory: at each step, the observer sees the world state and the executed joint action.
After the first $t$ joint steps, the observer has the partial trajectory
\[
\ObsSeq{t} = \{(\StateVar_\tau, \JointActionVar_\tau)\}_{\tau=1}^{t},
\]
where $\StateVar_\tau$ is the world snapshot at step $\tau$ and $\JointActionVar_\tau = (\JointActionVar_\tau^1,\dots,\JointActionVar_\tau^\NumAgents)$ is the executed joint action.
The online input from the environment is only this sequence.
The recognizer also has access to the learned policy $\PolicyTheta$ trained offline by MAGR-BB and uses it as a learned scoring model for candidate hypotheses.
We do not assume access to the agents' true private policies; $\PolicyTheta$ is the learned scoring model described in Section~\ref{sec:method:policy}.
The recognition task is to maintain the top-ranked hypotheses that best explain $\ObsSeq{t}$ as more steps are observed, ranked by a learned-policy score defined in Section~\ref{sec:method:scoring}.

\subsection{Goal- and Team-Conditioned Multi-Agent Policy}
\label{sec:method:policy}

The conditioning enters through an observation builder for candidate teams and goals.
The builder concatenates four fixed-dimensional feature blocks: domain-state features, self-identity features, team-membership features, and goal features.
Given a world snapshot, an acting agent $i$, a candidate team $T$, and a candidate team goal $g$, it constructs
\[
\CounterObs{t}{i}{T,g} =
\big[
\StateFeats{\StateVar_t},
\SelfEmbed{i},
\TeamEmbed{T},
\GoalEmbed{g}
\big].
\]
The first term summarizes the current world snapshot; the next two identify the acting agent and the candidate team via a binary team mask; the last encodes the candidate goal.
We call $\CounterObs{t}{i}{T,g}$ a \emph{counterfactual observation} because it keeps the world snapshot fixed while substituting candidate team and goal features.
Section~\ref{sec:experiments:setup} reports the concrete Blocksworld instantiation.

The teammate-action context $\ActCtx_t^i$ records the actions already chosen by earlier teammates at the same joint step.
We use a single shared policy parameterized by $\PolicyParams$,
\[
\PolicyTheta(\JointActionVar_t^i \mid \CounterObs{t}{i}{T,g}, \ActCtx_t^i),
\]
where $T \subseteq \AgentSet$ is a candidate team containing agent $i$, and $g$ is a candidate goal for that team.
Here $\JointActionVar_t^i$ is agent $i$'s executed action from the recorded joint action, not an observer decision or a full joint action.
The policy is goal-conditioned because the candidate goal $g$ enters the network input, and team-conditioned because an explicit team-membership mask for $T$ also enters that input.
A single shared network therefore covers all candidate $(T,g)$ pairs, and recognition scores hypotheses by reusing that same policy for each required action-probability query.
This design differs from GRAQL~\citep{amado2022graql}, which learns a separate policy or value function for each candidate goal in the single-agent case.
Training and recognition use the same builder.
Training fills in the true team and goal sampled by the environment, while recognition fills in hypothesized teams and goals.

We implement the shared policy with a Transformer encoder~\citep{vaswani2017attention} that maps the counterfactual observation to a distribution over the domain action space.
Architectural choices, token layout, and hyperparameters are domain-specific and reported in Section~\ref{sec:experiments:setup}.

We model within-team coordination with \emph{team-autoregressive decoding}.
At each joint step, teammates act in a fixed within-team order (ascending agent index within each team).
When predicting agent $i$'s action, the policy receives $\ActCtx_t^i$.
This turns a team's same-step joint action into an ordered product of conditional action probabilities and lets recognition score each observed action while conditioning on earlier teammate actions.

We train the actor with PPO~\citep{schulman2017ppo} under CTDE: during training only, the critic receives centralized information, while recognition later queries the learned actor as the scoring model.
The curriculum is stage-based: each stage fixes a goal-length range and an initial-state scrambling level, and later stages use longer goals and harder scrambles to handle sparse reward in the longer-goal regime.
Section~\ref{sec:experiments:setup} reports the exact schedule.

\subsection{Counterfactual Scoring}
\label{sec:method:scoring}

Counterfactual scoring evaluates how well a candidate team-goal hypothesis explains the recorded actions.
For a recorded state-action pair $(\StateVar_\tau,\JointActionVar_\tau)$, recognition keeps the observed world state and executed joint action unchanged and varies only the hypothesized team and goal used to build the policy input.

A concrete example illustrates the operation.
In the multi-agent Blocksworld benchmark of Section~\ref{sec:experiments}, suppose agents $1$ and $2$ both move at the same joint step, with agent $1$ clearing block $b$ and agent $2$ holding block $a$.
To score the hypothesis ``agents $1$ and $2$ form one team building the stack $a$ on $b$'' for agent $2$, we set agent $2$'s team mask to mark both agents $1$ and $2$ as members of the candidate team, insert agent $1$'s observed action into the teammate-action context because agent $1$ is earlier in the within-team order, and encode $a$-on-$b$ in the goal features.
We then query $\PolicyTheta$ for the probability of the action that agent $2$ actually executed.
A different hypothesis, for example ``agents $2$ and $3$ form a team building $a$ on $b$'', changes only the policy input used for scoring: the team mask, the teammate-action context, and the goal encoding.
The observed world state and the executed actions stay fixed.

For a given hypothesis $\HypothesisVar$, let $\CounterObs{\tau}{i}{\HypothesisVar}$ denote the policy input built by keeping the recorded state $\StateVar_\tau$ fixed and inserting the team and goal features specified by $\HypothesisVar$ for agent $i$.
Let $\ActCtx_\tau^i(\HypothesisVar)$ denote the same-step teammate-action context formed by selecting, from the recorded joint action $\JointActionVar_\tau$, the actions of agents that $\HypothesisVar$ places in the same team as $i$ and earlier in the fixed within-team order.
Thus $\HypothesisVar$ changes which recorded actions enter the context, but it does not change the recorded actions themselves.
Some candidate hypotheses make the recorded action infeasible in the recorded state.
In Blocksworld, for example, a hypothesis can assign agent $i$ to a goal whose block subset does not contain the block used by the recorded action.
For feasible recorded actions, we score the action by the clipped log-likelihood under the policy:
\begin{equation}
\label{eq:score}
\ScoreFn{t}{\HypothesisVar} =
\sum_{\tau=1}^{t}\sum_{i=1}^{\NumAgents}
\log \max\!\Big(
\PolicyTheta\big(\JointActionVar_\tau^i \mid \CounterObs{\tau}{i}{\HypothesisVar}, \ActCtx_\tau^i(\HypothesisVar)\big),
\ScoreEps
\Big),
\end{equation}
where $\ScoreEps>0$ is a small constant that prevents undefined logarithms when the policy assigns zero probability to the executed action.
We call $\ScoreFn{t}{\HypothesisVar}$ in Eq.~\ref{eq:score} the \emph{hypothesis score} of $\HypothesisVar$ at time $t$, and use this single name throughout the paper.
If the executed action is infeasible under a candidate hypothesis's action constraints, the corresponding term contributes $\log \ScoreEps$ rather than $-\infty$, and the hypothesis is heavily but finitely penalized.
The likelihood terms score how well a hypothesis explains the observed actions.

\begin{definition}[Non-competitive score]
\label{def:noncompetitive}
For a fixed observed trajectory and partition $\PartitionVar$, a learned-policy score is non-competitive if each agent-level term for $i\in\TeamVar{k}$ and any terminal penalty depend on $\TeamVar{k}$ and $\GoalVar{k}$, but not on other teams' candidate goals.
This is a score condition, not a requirement that the physical domain use disjoint workspaces.
\end{definition}

The observer still ranks complete hypotheses that assign every agent to a team and one goal per team; the condition only requires that each team's likelihood contribution be scored without reading the candidate goals of the other teams.
At the final observed state, we also apply a finite terminal-consistency penalty: if the trajectory terminates without truncation and a candidate team goal is not satisfied, we subtract $\TermPenalty>0$ once from that team's local score.
This penalty does not replace the action likelihood; it adds a finite correction against hypotheses whose predicted goal was never achieved.
The values of $\ScoreEps$ and $\TermPenalty$ are reported in Section~\ref{sec:experiments:setup}.
The next subsection groups these likelihood terms into the table entries that branch-and-bound reuses and prunes.

\subsection{Factorized Branch-and-Bound Recognition}
\label{sec:method:recognition}

At recognition time, we do not rescore every complete partition-goal hypothesis from scratch after each new observation.
Instead, we maintain cumulative scores for local team-goal pairs and reuse them across all complete hypotheses that contain that pair.
For a partition $\PartitionVar=\{\TeamVar{1},\dots,\TeamVar{\NumTeams}\}$, a team index $k\in\{1,\dots,\NumTeams\}$, and a candidate goal $g\in\GoalSpace{k}$ for team $\TeamVar{k}$, we define the local team-goal score as
\begin{equation}
\label{eq:local_score}
\ScoreFn{t}{\TeamVar{k}, g} =
\sum_{\tau=1}^{t}\sum_{i\in \TeamVar{k}}
\log \max\!\Big(
\PolicyTheta\big(\JointActionVar_\tau^i \mid \CounterObs{\tau}{i}{\TeamVar{k}, g}, \ActCtx_\tau^i(\TeamVar{k},g)\big),
\ScoreEps
\Big).
\end{equation}
The inner sum ranges over the agents assigned to team $\TeamVar{k}$, and the context term $\ActCtx_\tau^i(\TeamVar{k},g)$ uses the same within-team order as the policy.
If the trajectory has terminated without truncation and goal $g$ is unsatisfied in the final state, we subtract $\TermPenalty$ once from this local score.
Under Definition~\ref{def:noncompetitive}, every term in Eq.~\ref{eq:score} that contains agent $i\in\TeamVar{k}$ depends only on $\TeamVar{k}$ and $\GoalVar{k}$, so grouping the terms by team yields the additive decomposition
\begin{equation}
\label{eq:additive}
\ScoreFn{t}{\PartitionVar,\GoalVar{1},\dots,\GoalVar{\NumTeams}}
= \sum_{k=1}^{\NumTeams} \ScoreFn{t}{\TeamVar{k}, \GoalVar{k}}.
\end{equation}
Eq.~\ref{eq:additive} is the central property: the cumulative score of a complete partition-goal hypothesis equals the sum of its team-goal local scores.
For a fixed partition, the recognizer stores one local score for each team slot and candidate goal.
When a new observation arrives, a scoring-stage bound can skip an entire partition.
For each partition that survives this test, the recognizer refreshes all local team-goal rows for that partition before computing its current bound.
Thus the method still scores local team-goal pairs for surviving partitions; it avoids constructing the Cartesian product of all team-goal choices unless those complete hypotheses can still enter the top-$k$ ranking.
At each observed step, the online procedure rebuilds a global min-heap of the current top-$k$ complete hypotheses from the current local-score table.
Once that heap is full, its smallest retained score is the floor that admissible bounds compete against.

For any partition $\PartitionVar$ whose local rows are available at step $t$, the partition upper bound is
\begin{equation}
\label{eq:partition_bound}
\UpperBoundFn{t}{\PartitionVar}=\sum_{k=1}^{\NumTeams}\max_{g\in \GoalSpace{k}} \ScoreFn{t}{\TeamVar{k}, g},
\end{equation}
which selects, for each team slot, the largest refreshed local score under that partition.
Because each per-slot maximum dominates any specific choice of $\GoalVar{k}$, the value $\UpperBoundFn{t}{\PartitionVar}$ upper-bounds every complete hypothesis consistent with $\PartitionVar$, by Eq.~\ref{eq:additive}.

Inside one surviving partition, we still need to search combinations that choose one goal per team.
We sort each team's candidate goals by local score in descending order, paying a per-partition sorting cost over the refreshed local rows.
These sorted lists do not enumerate complete hypotheses; each is a per-team list used to generate complete hypotheses in best-first order.
The goal at index $j_k=1$ in team $k$ is the team's currently best candidate.
An index tuple $(j_1,\dots,j_{\NumTeams})$ then denotes the complete hypothesis that uses the $j_k$th goal in team $k$'s sorted list, and its score is the sum of the selected local scores.
Increasing any coordinate $j_k$ moves down team $k$'s sorted list, so the local score of that team can only stay the same or decrease, and so can the sum.
A best-first max-heap over index tuples therefore admissibly searches the goal-combination grid: the heap top is always the best unexplored complete hypothesis, and each emitted tuple creates only its one-coordinate successors.
Once the heap top falls at or below the current floor, no remaining tuple in this partition can enter the ranking.

Algorithm~\ref{alg:magrbb} summarizes the online procedure: a scoring-stage prune skips refresh when the stale partition bound cannot beat the floor, a partition-level prune drops dominated partitions after refresh, and a local-level prune halts best-first tuple expansion once the heap top falls at or below the floor.
We say that the recognizer \emph{emits} a candidate whenever it constructs one complete partition-goal hypothesis to compare against the heap.
Stale bounds come from the current local table (zero before any observation), are replaced after refresh, and remain valid because refresh only decreases local scores.

\begin{algorithm}[!t]
\small
\setlength{\baselineskip}{0.95\baselineskip}
\algrenewcommand{\algorithmicindent}{1em}
\caption{MAGR-BB online recognition (Full B\&B variant).}
\label{alg:magrbb}
\begin{algorithmic}[1]
\Require Candidate partitions $\PartitionSet$, candidate goal libraries $\{\GoalSpace{k}\}$, top-$k$ size $\RankLimit$, observed trajectory arriving step by step.
\State Initialize $\ScoreFn{0}{\TeamVar{},g}=0$ and set stale bounds by Eq.~\ref{eq:partition_bound}.
\For{each new joint step $(\StateVar_t,\JointActionVar_t)$}
  \State Initialize empty min-heap $\TopKHeap$ and set $\Floor\gets-\infty$ for this step.
  \State Sort partitions in $\PartitionSet$ by stale upper bound $\StaleUpperBoundFn{t}{\PartitionVar}$ in descending order.
  \For{each partition $\PartitionVar$ in sorted order}
    \If{\label{algline:stale-test}$\TopKHeap$ is not full \textbf{or} $\StaleUpperBoundFn{t}{\PartitionVar} > \Floor$}
      \State \label{algline:refresh-scores}Refresh all $\ScoreFn{t}{\TeamVar{k},g}$ for this partition using Eq.~\ref{eq:local_score}.
      \State \label{algline:partition-bound}Compute $\UpperBoundFn{t}{\PartitionVar}=\sum_k\max_{g} \ScoreFn{t}{\TeamVar{k},g}$.
      \If{\label{algline:partition-test}$\TopKHeap$ is not full \textbf{or} $\UpperBoundFn{t}{\PartitionVar} > \Floor$}
        \State \label{algline:init-tuples}Sort each team's goals by local score, initialize $\GoalTupleQueue$ with $(1,\dots,1)$, and initialize $\textit{seen}\gets\{(1,\dots,1)\}$ (local to this partition iteration).
        \State $\textit{emitted}\gets0$.
        \While{\label{algline:local-loop}$\GoalTupleQueue$ is non-empty, $\textit{emitted}<\RankLimit$, and either $\TopKHeap$ is not full or the best tuple bound in $\GoalTupleQueue$ exceeds $\Floor$}
          \State \label{algline:emit-tuple}Pop the best tuple, emit its complete hypothesis, and insert it into $\TopKHeap$ if it belongs in the current top-$k$ set.
          \State \label{algline:expand-tuples}For each one-coordinate successor $\sigma$ of the popped tuple with $\sigma\notin\textit{seen}$, add $\sigma$ to $\textit{seen}$ and push $\sigma$ onto $\GoalTupleQueue$.
          \State Increment $\textit{emitted}$.
        \EndWhile
      \EndIf
    \EndIf
    \State If $\TopKHeap$ is full, update $\Floor$ to the current smallest score in $\TopKHeap$.
  \EndFor
  \State Store each refreshed $\UpperBoundFn{t}{\PartitionVar}$ as that partition's stale bound for later observed steps.
  \State Output $\TopKHeap$ as the ranking for step $t$.
\EndFor
\State \Return the latest ranking.
\end{algorithmic}
\end{algorithm}

\begin{proposition}[Additive score and admissible pruning]
\label{prop:additive_pruning}
Under Definition~\ref{def:noncompetitive}, Eq.~\ref{eq:additive} holds for every complete hypothesis.
The scoring-stage, partition-level, and local-level pruning tests in Algorithm~\ref{alg:magrbb} cannot discard a complete hypothesis whose score exceeds the current top-$k$ floor.
\end{proposition}

\begin{proof}
Definition~\ref{def:noncompetitive} lets us group Eq.~\ref{eq:score} by the disjoint teams, giving Eq.~\ref{eq:additive}.
For any partition, each $\max_{g\in\GoalSpace{k}}\ScoreFn{t}{\TeamVar{k},g}$ dominates any chosen $\GoalVar{k}$, so Eq.~\ref{eq:partition_bound} upper-bounds all complete hypotheses under that partition.
Zero-table bounds are valid, and refreshes append only terms $\le 0$, so stale bounds remain upper bounds.
Within a partition, sorted goal lists are non-increasing; once the best unexplored tuple is at or below the floor, no remaining tuple can enter the top-$k$ set.
\end{proof}

Since the score function and the local team-goal scores are the same in the exhaustive baseline and in every pruned variant, the pruned variants discard only hypotheses whose bound cannot exceed the current floor.
For steps with a strict score gap at the top-$k$ boundary, they return the same top-$k$ ranking as exhaustive search under the same score; if multiple hypotheses tie at that boundary, they return a score-equivalent ranking up to those tied hypotheses.

The factorized recognizer therefore spends its work on local team-goal tables and the admissible bounds that certify the same answer (Section~\ref{sec:experiments:ablations}).

\section{Experiments}
\label{sec:experiments}
\subsection{Experimental Setup}
\label{sec:experiments:setup}

At step $t$, the recognizer has consumed the first $t$ joint state-action pairs of the recorded trajectory.
We ask two questions.
First, whether the learned policy is accurate enough to serve as a behavior model for recognition.
If the policy assigns very low probability to the true actions even under the true team-goal hypothesis, every hypothesis sits near the same low-probability floor and the recognizer becomes uninterpretable.
Second, how much computation our pruning removes while preserving the reported top-ranked hypothesis and final top-$10$ list.
We measure recognition quality with final top-$1$ team accuracy, goal accuracy, joint accuracy, and identification latency.
We measure computation with cumulative runtime, score-table refreshes, partition visits, and goal-tuple emissions.

\noindent\textbf{Domain.}
Blocksworld is a classical planning domain consisting of blocks stacked into towers on a table; the planning problem is to transform an initial configuration into a target configuration by moving one block at a time~\citep{slaney2001blocksworld}.
The benchmark is a multi-agent Blocksworld environment instantiated with two hidden teams of size two.
Each team acts in its own seven-block workspace.
We use \emph{disjoint workspaces} as a controlled instantiation of Definition~\ref{def:noncompetitive}: they make the additive score factorization transparent and remove resource-conflict confounds, but they are not part of the MAGR-BB definition.
Goals are ordered \emph{stacks} of length 2--4: a stack of length $\ell$ is a sequence of $\ell$ blocks that must end up placed one on another in the specified order.
The $\ell-1$ adjacent on-relations define the visible tower.
The benchmark also stores each candidate as a full-state support assignment over all seven workspace blocks, so the unsatisfied-relation count checks that the bottom tower block and every block outside the tower remain on the table.
With seven blocks per workspace and stacks of length up to four, each team-slot goal library contains $|\GoalSpace{k}|=\sum_{\ell=2}^{4}7!/(7-\ell)!=42+210+840=1092$ candidate goals.
With six indexed partitions, exhaustive ranking considers $6\times 1092^2=7{,}154{,}784$ complete partition-goal hypotheses per observed step.

\noindent\textbf{Observation and action space.}
Each agent observation has dimension $166$ and concatenates domain-state features (on-relations, clear predicates, holding status), agent-identity features, a team-membership mask, and the goal feature that encodes the candidate stack as pairwise on-relations.
The discrete action space contains $99$ actions covering \texttt{noop}, \texttt{pickup}, \texttt{putdown}, \texttt{stack}, and \texttt{unstack} primitives over seven blocks.
The two team slots are labeled by workspace, so the three unordered pairings of four agents into two pairs yield six indexed partition hypotheses.

\noindent\textbf{Policy architecture and training.}
The shared policy is a Transformer encoder over seven block tokens, four agent tokens, and one learned \texttt{[CLS]} token, with two attention layers, four heads, and hidden size $256$.
The actor predicts the operation type from \texttt{[CLS]} and the block arguments from the block tokens, then assembles the two factors into logits over the $99$ actions.
We train with PPO under CTDE (clip $0.2$, learning rate $10^{-5}$, batch size $256$, four epochs per rollout, $\Discount=0.99$) using eight parallel environments and team-autoregressive decoding, under a six-stage length-and-difficulty curriculum with mastery gates.
The saved configuration uses CPU execution for training and recognition; the reported recognition ablations use no GPU.

\noindent\textbf{Recognition protocol.}
Recognition uses the score $\ScoreFn{t}{\HypothesisVar}$ defined in Eq.~\ref{eq:score} and retains the top-$10$ hypotheses after every observed step.
The score uses $\ScoreEps=10^{-10}$ for the log-likelihood clip and $\TermPenalty=2.0$ for the terminal inconsistency penalty.
We replay trajectories under five rollout seeds and four \emph{action-noise} levels $\{0.0,0.05,0.1,0.2\}$.
At each joint step, a Bernoulli draw with probability $p$ decides whether to perturb the step; if perturbed, each agent's action is sampled independently from that agent's current valid-action mask.
All reported recognition numbers average over the five trajectories at each noise level.

\noindent\textbf{Compared variants.}
We compare six variants of the same factorized recognizer.
\emph{Factorized Exhaustive} performs no pruning.
\emph{Scoring Prune}, \emph{Partition B\&B}, and \emph{Local B\&B} enable one pruning stage at a time.
\emph{Ranking B\&B} combines partition and local ranking bounds, and \emph{Full B\&B} enables all three mechanisms (Algorithm~\ref{alg:magrbb}).
Every variant uses the same local team-goal scores and the same complete-hypothesis score, so the comparison isolates computational savings under a fixed recognition objective.
We sum per-step runtime and search counters (score-table refreshes, partition visits, and goal-tuple emissions) over each trajectory before averaging over seeds, and we report the work at the final observed step separately.
Five trajectories per noise level is a small sample, so we treat these results as in-domain evidence rather than as a broad benchmark claim.

\subsection{Policy Performance}
\label{sec:experiments:policy}

The final checkpoint achieves $98.44\%$ episode success rate (both teams complete their goals) and $99.22\%$ per-team success rate, with success at or above $97.62\%$ across stack lengths $2$--$4$ and an average of $0.023$ unsatisfied goal relations per episode.
The true team-goal hypothesis therefore separates cleanly from wrong hypotheses, which lets the bounds in Section~\ref{sec:method:recognition} prune effectively.

\subsection{Recognition Results}
\label{sec:experiments:recognition}

All six variants agree on the top-$1$ hypothesis at every observed step and on the final top-$10$ list across all trajectories and noise levels, recovering the correct team partition and joint team-goal assignment in every condition (Table~\ref{tab:accuracy_noise}).
Team identity resolves after $1.6$ observed steps on average; full joint team-goal identification occurs after $6.4$--$7.2$ steps depending on noise.

\begin{table}[!htbp]
\centering
\small
\begin{tabular}{lccccccc}
\toprule
Noise & Team acc. & Goal acc. & Joint acc. & Team lat. & Goal lat. & Joint lat. & Mean steps \\
\midrule
0.00 & 1.00 & 1.00 & 1.00 & 1.6 & 6.4 & 6.4 & 7.0 \\
0.05 & 1.00 & 1.00 & 1.00 & 1.6 & 6.4 & 6.4 & 7.0 \\
0.10 & 1.00 & 1.00 & 1.00 & 1.6 & 6.8 & 6.8 & 7.4 \\
0.20 & 1.00 & 1.00 & 1.00 & 1.6 & 7.2 & 7.2 & 8.4 \\
\bottomrule
\end{tabular}
\caption{
Final top-$1$ accuracy and identification latency, identical across variants (one row per noise level).
``Lat.'' is the first observed step at which the top-ranked hypothesis becomes perfectly correct for the indicated target.
``Mean steps'' is the mean replayed trajectory length.
}
\label{tab:accuracy_noise}
\end{table}

All six variants follow the same per-step top-$1$ accuracy curve because their serialized assignments match at every observed step: branch-and-bound changes computation, not the top-ranked hypothesis.

\subsection{Search Burden and Ablations}
\label{sec:experiments:ablations}

Because the top-$1$ accuracy curves are identical across variants, the empirical question is where the computation goes.
We therefore focus on wall-clock runtime and cumulative workload patterns.

\begin{figure}[!htbp]
\centering
\includegraphics[width=\linewidth]{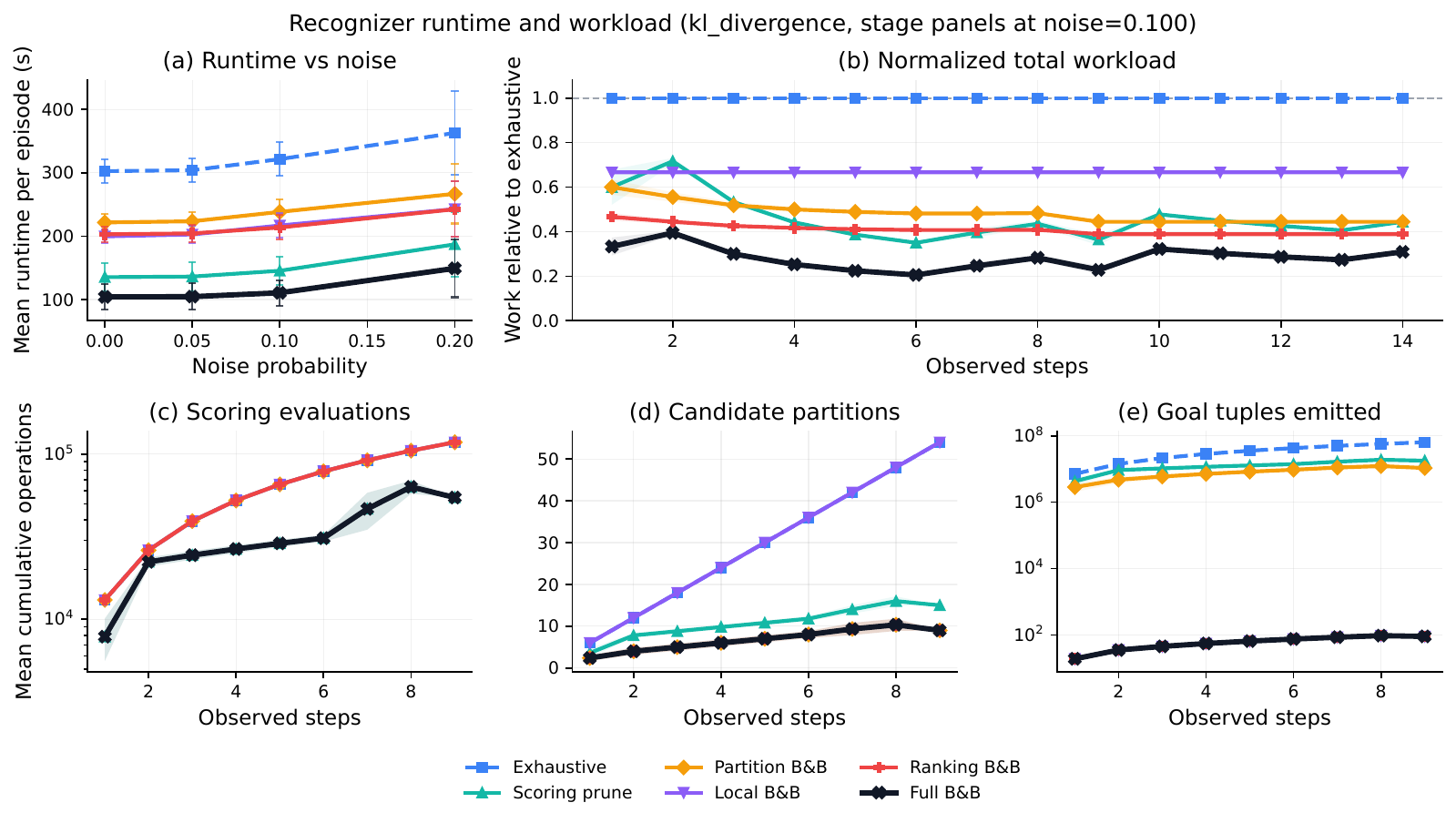}
\caption{
Runtime and workload across the six variants.
Top row: cumulative runtime versus action noise and normalized total work relative to Factorized Exhaustive.
Bottom row: cumulative work at action-noise $0.1$ split into score-table refreshes, partition visits, and goal-tuple emission.
Means over five trajectories; bands show standard errors.
}
\label{fig:runtime_workload}
\end{figure}

Figure~\ref{fig:runtime_workload} reports cumulative runtime and workload across noise levels.
The ordering is stable across noise levels: Full B\&B is always best, and Scoring Prune captures most of the savings.
Ranking B\&B (partition plus local pruning, no scoring skip) is much weaker than Full B\&B; combining ranking pruning with scoring-stage skips is what produces the full gain.
Partition B\&B reduces the number of visited partitions but leaves score updates unchanged because it acts only after scoring has already refreshed a partition.
Local B\&B mostly reduces goal-tuple emission rather than score-table work.
Scoring Prune and Full B\&B are the only variants that materially reduce cumulative scoring evaluations, which explains most of the runtime gain.

Table~\ref{tab:full_vs_exhaustive} isolates Full B\&B against the exhaustive baseline.
We report two regimes rather than four rows because noise $0.0$, $0.05$, and $0.10$ match at the level of abstraction reported here, while noise $0.20$ separates from that regime.
At the final observed step, Full B\&B visits one candidate partition rather than all six and emits $10$ complete goal tuples rather than $7.15$M.
Over the full replay, relative to Factorized Exhaustive, it cuts cumulative score updates by about half at low noise and by $39\%$ at the highest tested noise level.
Cumulative tuple emission falls from $50.1$M--$52.9$M to $85$--$89$ for noise $\le 0.10$, and from $60.1$M to $99$ at noise $0.20$.

\begin{table}[htb]
\centering
\small
\setlength{\tabcolsep}{4pt}
\begin{tabular}{lccccc}
\toprule
Noise & \shortstack{Final partition\\visits} & \shortstack{Final tuple\\emissions} & \shortstack{Cumulative tuple\\emissions} & \shortstack{Cumulative score\\updates} & \shortstack{Runtime\\speedup} \\
\midrule
$\le 0.10$ & $6 \to 1$ & $7.15$M $\to 10$ & $50.1$--$52.9$M $\to 85$--$89$ & $91.7$--$97.0$k $\to 46.7$--$50.2$k & $2.91\times$ \\
$0.20$ & $6 \to 1$ & $7.15$M $\to 10$ & $60.1$M $\to 99$ & $110.1$k $\to 67.3$k & $2.43\times$ \\
\bottomrule
\end{tabular}
\caption{
Full B\&B versus Factorized Exhaustive at the final observed step and over the full replay; arrows point Exhaustive $\to$ Full B\&B, averaged over five trajectories.
``Final partition visits'' counts candidate partitions in the last ranking step.
``Runtime speedup'' is the ratio of Exhaustive to Full B\&B cumulative runtime.
}
\label{tab:full_vs_exhaustive}
\end{table}

Noise mainly weakens the scoring-stage prune: at noise $0.2$, the floor separates later, so the recognizer must refresh more local tables before it can certify that a partition is dominated, which is why cumulative score updates rise for the pruning variants at the highest tested noise level.
Even so, Full B\&B emits at most $99$ cumulative goal tuples and remains the fastest variant at every tested noise level, supporting the central claim that online recognition reduces to maintaining local score tables and applying admissible bounds rather than constructing millions of complete hypotheses.

\FloatBarrier

\section{Related Work}
\label{sec:related_work}

Single-agent model-free recognizers rank candidate goals by behavior-model fit: GRAQL~\citep{amado2022graql} learns one policy or value function per goal, \citet{Fang2023}'s deep-RL recognizer extends this idea to continuous domains, and GRNet~\citep{grNet} is a neural goal recognizer.
These methods target one observed actor, whereas MAGR-BB trains one shared policy conditioned on candidate team and goal, so one network scores every candidate team-goal pair.

The multi-agent recognizers closest to our setting rely on explicit plan, action, or intention structures.
\citet{banerjee2010multiagent} formalize multi-agent plan recognition, prove that the variable-agent search is NP-complete, and solve it with Knuth's Algorithm~X over symbolic plan hypotheses.
\citet{zhuo2012action,zhuo2019recognizing} recognize multi-agent plans from a hand-authored action model.
\citet{argenta2015maprap,argenta2016probabilistic} extend recognition-as-planning to jointly recover teams, goals, and plans.
\citet{michaeldann2023multiagent} combine online goal recognition with MCTS-based intention progression.
Together, these approaches use plan libraries, hand-authored action models, recognition-as-planning encodings, or BDI intention structures; MAGR-BB instead scores team-goal hypotheses with a learned multi-agent policy.

Coalition-structure generation~\citep{rahwan2007coalition} inspires our partition layer, but it optimizes coalition utilities rather than goal hypotheses; our bounds come from additive learned-policy scores.
Planning-based recognizers filter candidate goals with costs~\citep{ramirez2009plan,ramirez2010probabilistic} or landmarks~\citep{pereira2020landmark}; MAGR-BB instead reuses local team-goal scores across prefixes.

\section{Conclusion}
\label{sec:conclusion}

We introduced MAGR-BB, a model-free recognizer for joint team and goal inference.
MAGR-BB scores local team-goal pairs with one conditioned Transformer and combines those scores with admissible B\&B under the non-competitive score condition.
On the controlled disjoint-workspace Blocksworld benchmark, it matches exhaustive top-$1$ decisions at every observed step and the final top-$10$ list while reducing online runtime by $2.43$--$2.91$.
The disjoint-workspace benchmark is a proof-of-concept instantiation of the score condition, not a requirement built into MAGR-BB.
Future work should test shared resources, partial observations, structured noise, and larger teams and goals.

\bibliographystyle{plainnat}
\bibliography{references}

\end{document}